\documentclass[11pt,a4paper]{article}
\usepackage{geometry}
\usepackage{graphicx}
\usepackage{bm,array}
\usepackage{comment}
\usepackage{caption}
\usepackage{subcaption}
 \usepackage{makecell}
 \usepackage{multirow}
 \usepackage{float}

\usepackage[normalem]{ulem}

\usepackage[pdfpagemode={UseOutlines},bookmarks=true,bookmarksopen=true,
bookmarksopenlevel=0,bookmarksnumbered=true,hypertexnames=false,
colorlinks,linkcolor={blue},citecolor={blue},urlcolor={red},
pdfstartview={FitV},unicode,breaklinks=true]{hyperref}
\hypersetup{urlcolor=blue, colorlinks=true}

\usepackage{setspace}
\usepackage{vmargin}
\setmarginsrb{ 1.0in}  
{ 0.6in}  
{ 1.0in}  
{ 0.8in}  
{  20pt}  
{0.25in}  
{9pt}  
{ 0.3in}  

\usepackage[round, sort,comma,authoryear]{natbib}

\usepackage{authblk}

\usepackage{amsmath}
\usepackage{amsfonts}
\usepackage{amssymb}

\newtheorem{theorem}{Theorem}

\newtheorem{proposition}{Proposition}

\newenvironment{proof}[1][Proof]{\noindent\textbf{#1.} }{\ \rule{0.5em}{0.5em}}

\usepackage{multirow}
\usepackage{hhline}
\usepackage{footnote}
\usepackage[ruled,vlined]{algorithm2e}
\usepackage{algorithmic}

\SetCommentSty{mycommfont}

\newcommand{\A}{\mathcal{A}}

\newcommand{\transpose}{{\mbox{\tiny T}}}

\newcommand{\tA}{\widetilde{\mathcal{A}}}
\newcommand{\cA}{{\mathcal{A}}}

\newcommand{\cG}{{\mathcal{G}}}
\newcommand{\cV}{{\mathcal{V}}}

\newcommand{\cS}{{\mathcal{S}}}

\newcommand{\cP}{{\mathcal{P}}}

\newcommand{\cT}{{\mathcal{T}}}

\newcommand{\bE}{{\mathbb{E}}}

\newcommand{\bI}{{\textbf{I}}}

\newcommand{\bM}{{\textbf{M}}}

\newcommand{\bQ}{\textbf{Q}}

\newcommand{\bq}{\textbf{q}}

\newcommand{\bU}{\textbf{U}}
\newcommand{\bD}{\textbf{D}}
\newcommand{\bH}{\textbf{H}}
\newcommand{\bZ}{\textbf{Z}}
\newcommand{\bA}{\textbf{A}}

\newcommand{\bL}{\textbf{L}}

\newcommand{\ba}{\textbf{a}}

\newcommand{\bb}{\textbf{b}}
\newcommand{\bh}{\textbf{h}}

\newcommand{\bpii}{\pmb{\pi}}

\newcommand{\bGm}{{\pmb{\Gamma}}}
\newcommand{\bgm}{{\pmb{\gamma}}}
\newcommand{\btheta}{{\pmb{\theta}}}

\newcommand{\bpi}{\pmb{\Pi}}

\usepackage{mathtools}

\usepackage{mathtools}

\newif\ifnotes\notestrue
\def\htien#1{}

\usepackage{caption}
\usepackage{subcaption}
\def\<#1>{\textit{#1}}

\begin{document}
	
\newcolumntype{C}{>{\centering\arraybackslash}p{4em}}

\title{\textbf{Estimation of Recursive Route Choice Models with Incomplete Trip Observations}}
\author[1]{Tien Mai}
\author[1]{The Viet Bui}
\author[2]{Quoc Phong Nguyen}
\author[3]{Tho V. Le}
\affil[1]{\it\small School of Computing and Information Systems, Singapore Management University }
	\affil[2]{\it\small School of Computing, National University of Singapore }
		\affil[3]{\it\small Department of Civil \& Mineral Engineering, University of Toronto }

\maketitle

\begin{abstract}
This work concerns the estimation  of recursive  route choice models in the situation that the trip observations are incomplete, i.e.,  there are  unconnected links (or nodes) in the observations.  A direct approach to handle this  issue would be intractable because enumerating all paths between  unconnected links (or nodes) in a real network is typically not possible. 
We exploit an \textit{expectation-maximization} (EM) method that allows to deal with the missing-data issue by  alternatively performing two steps of  sampling the missing segments in the observations and solving maximum likelihood estimation problems. 
Moreover, observing that the EM method would be expensive,
we propose a new estimation method based on the idea that the choice probabilities of unconnected link observations can be exactly computed by solving systems of linear equations. We further design a new algorithm, called as \textit{decomposition-composition} (DC), that helps  reduce the number of systems of linear equations to be solved and speed up the estimation. We compare our proposed algorithms with some standard baselines using a dataset from a real network, and show that  the DC algorithm outperforms the other approaches in recovering missing information in the observations. Our methods work with most of the recursive route choice models proposed in the literature, including the recursive logit, nested recursive logit, or discounted recursive models.


\end{abstract}


\section{Introduction}
\label{sec:intro}

Travel demand modeling plays a vital role in transportation planning, operations, and management. In travel demand modeling, route choice is a process that associates the behaviors of selecting a route to travel from an origin to a destination. Such a route choice model can be used to assess  travelers’ preferences regarding characteristics of different routes and travellers (e.g. travel time, travel cost, number of crossings, or age, gender, and trip purpose of the travellers). Route choice models can also be used to predict traffic flows and can have  applications in  traffic simulation \citep{osorio2015computationally} or network pricing \citep{zimmermann2021strategic}. 
The parameters of a route choice model are often estimated using path observations collected by, for instance, Global Positioning System (GPS) devices equipped on vehicles. It is commonly known that such trip observations would be incomplete due to several reasons, for example, the GPS signals are intermittent or some parts of the GPS data are censored from the analyst due to privacy concerns. This missing data issue is critical for the estimation of route choice models  but to the best of our knowledge,  this issue has not been properly investigated in the context of route choice modeling. We address the issue in this paper.  

We focus on link-based recursive route choice models  \citep{FosgFrejKarl13,MaiFosFre15} due to their several advantages, namely,  the estimation does not require sampling of route choice alternatives and  the models can be consistently estimated and  easy to predict. Such a recursive model is based on the Rust's dynamic discrete choice framework \citep{RUST1987} where a path choice is modelled as a sequence of link choices and the parameters  are estimated by combining maximum likelihood estimation (MLE) and dynamic programming \citep[see][ for a tutorial]{zimmermann2020tutorial}. A naive approach to deal with the missing-data issue under recursive  models is to ignore the unconnected link observations  and solve the MLE problem only based on connected ones. This approach can be done straightforwardly using  existing estimation algorithms \citep{FosgFrejKarl13,MaiFosFre15,MaiBasFre15_DeC}. However, unconnected segments in the observations would contain valuable information about traveller route choice behavior  and would be important for the  estimation. In this paper, we aim at developing new estimation algorithms that allow us to efficiently and practically recover information from these incomplete segments. 

\noindent\textbf{Our contributions:} We make a number of contributions in this work. We first exploit the  \textit{expectation-maximization} (EM) method to deal with the missing-data issue. This method has been widely used to estimate statistical models with unobserved variables  \citep{dempster1977maximum,mclachlan2019finite} and we show how this scheme can be used  in the context of route choice modeling under recursive models. More precisely, we design an EM algorithm that alternatively performs two steps where the first step is to compute an expected log-likelihood function  by sampling paths to fill unconnected segments in the observations, and the second step maximizes that expected log-likelihood function. We further show that, under the recursive logit  model \citep{FosgFrejKarl13}, the expected log-likelihood function is concave, thus the second step can be done efficiently by convex optimization.     

Moreover, observing that the EM algorithm would be expensive to perform as  it requires solving several MLE problems, we further propose a new method to directly compute the probabilities of unconnected path segments. That is, we show that the probability of unconnected links can be obtained by solving one \textit{system of linear equations}. This result allows us to compute the likelihood function of the incomplete observations where the missing segments are explicitly taken into consideration. Furthermore, to reduce the number of systems of linear equations to be solved  and speed up the estimation process,  we develop  a new algorithm, called as the \textit{decomposition-composition} (DC) method. The idea is to \textit{decompose} each linear system above into two parts and \textit{compose} the common parts to reduce the number of linear systems to be solved. The main advantage of our DC algorithm is that we only need to solve \textit{only one system of linear equations} to obtain all the choice probabilities of unconnected segments on trips that share the same destination.  We also show that such a system of linear questions always has \textit{a unique solution}. 

We test our proposed algorithms with two popular recursive route choice models, i.e., the recursive logit (RL) \citep{FosgFrejKarl13} and nested recursive logit (NRL) models, using a dataset collected from a real network. Our numerical experiments clearly show that our DC algorithm outperforms the EM and other baseline approaches in recovering missing information from observations. Our numerical results are based on the RL and NRL models, but our method can be applied to other existing recursive models in the literature, e.g., the recursive network multivariate extreme value \citep{Mai_RNMEV}, discounted RL \citep{oyama2017discounted}, or stochastic time-dependent RL \citep{mai2021RL_STD}. 

\noindent \textbf{Literature reviews:}
The literature on route choice modeling covers path-based and link-based recursive models.  Path-based models \citep[see][for a review]{Prat09} rely on sampling of paths, thus the parameter estimates would depend on the sampling methods. Even when corrections are added to the choice probabilities to achieve consistent estimates, the prediction is not easy to perform. In contrast, the link-based recursive models are built based on the dynamic discrete choice framework \citep{RUST1987} and they are basically equivalent to discrete choice models based on the sets of all possible paths. Link-based recursive models have several advantages, making them attractive recently, for instance, the models can be consistently estimated and can be quickly and easily used for prediction.  The first recursive model (i.e. the RL model) is proposed by \cite{FosgFrejKarl13}. Several recursive models have been developed afterwards to deal with, for instance,  correlations between  path utilities \citep{MaiFosFre15,Mai_RNMEV,MaiBasFre15_DeC}, dynamic networks \citep{de2020RL_dynamic}, stochastic time-dependent networks \citep{mai2021RL_STD}, or discounted behavior \citep{oyama2017discounted}, with applications in, for instance, traffic management \citep{BailComi08,Melo12} or network pricing \citep{zimmermann2021strategic}. We refer the reader to  \cite{zimmermann2020tutorial} for a comprehensive review.  

With the development of technology and the ease of accessing the Internet, onboard devices, e.g., GPS, are more commonly used for several purposes, including rote navigation for taxis, buses, trucks, passenger cars, utility vehicles, mobile phones, etc. From operations and management perspectives, GPS data can be utilized to understand route choice behaviors. In fact, GPS has been used for collecting data for route choice modelling from a very early time \citep{wolf1999accuracy}. Some other early studies used GPS data for modelling the route choice of vehicles \citep{gilani2005automatically}, cyclists and pedestrians \citep{krizek2007detailed,dill2008understanding,menghini2010route,broach2012cyclists}, transit \citep{oliveira2010improving}, and commuting \citep{papinski2009exploring}. While collecting data by GPS brings many advantages, it is frequent that GPS data is missing or incomplete, leading to difficulties for analysis \citep{wang2011challenges,shen2014review}. Consequently, researchers have developed various approaches to overcome this issue. For example,  \cite{smith2003exploring} introduce a preliminary analysis of techniques for imputing missing data. The results revealed that the statistical technique outperformed the heuristics technique in terms of accuracy and performance.  \cite{melville2004experiments} used an ensemble method which is available under the DECORATE package and found DECORATE is more robust to missing data than the other two methods. Some researchers combined map-matching \citep{BierFrej08} or questionnaire surveys \citep{Ramos2015thesis} which provide ground truth and improve the quality of the GPS data.  We refer the reader to \cite{shen2014review} for a review of GPS travel surveys, applications, and data processing. In the context of route choice modeling, to deal with the missing-data issue, to the best of our knowledge, existing studies only focus on  data collection, i.e., improving GPS data and the map-matching, and we seem to be the first to take the estimation of route choice models with missing data into consideration.

Our work also closely relates to the EM framework \citep{dempster1977maximum,gentle1998algorithm,arcidiacono2011conditional}, which has been popularly used for estimating statistical models with latent variables or unknown data observations.  It should be noted that an EM algorithm  
does increase the likelihood value of the  observed data after each iteration,  but there is no guarantee  that the sequence converges to a maximum likelihood estimator \citep{gentle1998algorithm}. In addition, the method requires  repeatedly solving MLE problems, which would be expensive, especially in the context of recursive route choice models since the estimation of these models involves solving several dynamic programs \citep{MaiFrejinger22}.

\noindent \textbf{Paper outline:} This paper is organized as follows. Section \ref{sec:RL models} presents the recursive models  and the missing-data issue. Section \ref{sec:EM}  presents our EM  and Section \ref{sec:DC} presents  the DC  algorithms. Section \ref{sec:experiments} reports numerical experiments, and finally, Section \ref{sec:concl} concludes the study and provides directions for future research.

\noindent \textbf{Notations:}
Boldface characters represent matrices (or vectors), and $a_i$ denotes the $i$-th element of vector $\ba$, $A_{i,j}$ denotes the $(i,j)$-th element of matrix $\bA$, and  $\bA_{:,j}$ denotes the $j$-th column of matrix $\bA$. 
 We use $[m]$, for any $m\in \mathbb{N}_+$, to denote the set $\{1,\ldots,m\}$.

\section{Recursive Route Choice Models and the Missing Data Issue}
\label{sec:RL models}
Consider a network $\cG = (\cA,\cV)$, where $\cA$ is the set of links  and $\cV$ is the set of nodes.
Moreover, an absorbing state is associated with each destination by extending the network
with dummy links $d$. The set of all links becomes  $\tA = \A \cup \{d\}$.
For each link $k\in \cA$, let $A(k)$ be the set of outgoing links from $k$. For two links  $k,a\in \tA$, 
the instantaneous utility $u(a|k) = v(a|k)+ \epsilon(a|k)$, where  $v(a|k)$ is the
deterministic part of the utility and $\epsilon(a)$ is a random term with zero mean. We assume that  $v(a|k)$,
for all $a\in A(k)$ to be negative for all links except the dummy link $d$.

For any link $k\in \A$, let  $V(k)$ be the expected maximum utility (i.e., the value function) from the sink node of $k$ to the destination. In a recursive route choice model \citep{FosgFrejKarl13,MaiFosFre15}, a  choice of route  is modelled as a sequence of link choices, in which at each link, the traveller selects an outgoing link by maximizing the sum of link utility and the expected maximum utility from that link to the destination. The selection of the next link can be written as in the following form
$a = \text{argmax}_{a'\in A(k)} \left\{ v(a'|k)+ V(a') + \epsilon(a'|k)\right\}$. The value function $V(k)$ thus can be computed recursively  as 
\begin{equation}
    V(k) =
    \begin{cases}
    \bE_\epsilon\left[\max_{a\in A(k)} \left(v(a|k)+V(a) + \epsilon(a|k)\right)\right] & \text{ if }k\in\cA \backslash  \{d\}\\
    1 & \text{ if } k=d
    \end{cases}
\end{equation}
Here, we note that the value function should depend on the destination $d$ but we omit the corresponding indices for notational simplicity. Different assumptions can be made for the random terms $\epsilon(a|k)$. If we assume that $\epsilon(a|k)$, $\forall a\in A(k)$,  are i.i.d Extreme Value Type I with scale $\mu>0$ and independent of everything in the network, the recursive model becomes the  recursive logit (RL) \citep{FosgFrejKarl13} in which the value function can be computed conveniently as
\begin{equation}\label{eq:bellman-RL}
V(k) =
\begin{cases}
\mu\ln\left(\sum_{a\in A(k)} \exp\left(\frac{1}{\mu}\left(v(a|k)+ V(a)\right)\right)\right)&\forall  k \in \cA\backslash \{d\}\\
1&\text{ if } k=d.
\end{cases}
\end{equation}
The above recursive formulation has an appealing property that it can be written as a system of linear equations. More precisely, if we define a vector $\bZ$ of size $|\tA|$ with elements $Z_k = \exp(V(k)/\mu)$, for all $k\in \tA$, a matrix $\bM$ of size $(|\tA|\times|\tA|)$ with elements $M_{ka} = \mathbb{I}(a\in A(k))\exp\left(v(a|k)/\mu\right)$, and $\bb$ is a vector of size $|\tA|$ with zero elements everywhere except $b_d=1$, then \eqref{eq:bellman-RL} can  be reformulated as a system of linear equations  as $\bZ = \bM\bZ +\bb$. Such a system can be solved quickly using an iterative algorithm implemented in some existing linear solvers \citep{MaiFosFre15,MaiBasFre15_DeC}.  

It is well-known that the  RL model retains the independence from irrelevance alternatives (IIA) property \citep{FosgFrejKarl13,MaiFosFre15}, which would not hold in some real-life applications. To relax this property, \cite{MaiFosFre15} propose the NRL model based on the assumption that  the random terms $\epsilon(a|k)$ have different scale parameters over links, i.e., $\epsilon(a|k) = \mu_k\epsilon'(a|k)$, $\forall k\in\cA$, where $\epsilon'(a|k)$ are i.i.d extreme value type I for all $k,a\in\tA, a\in A(k)$. The value function in the NRL model  can be defined recursively as  
\begin{equation}\label{eq:bellman-NRL}
V(k) =
\begin{cases}
\mu_k\ln\left(\sum_{a\in A(k)} \exp\left(\frac{1}{\mu_k}\left(v(a|k)+ V(a)\right)\right)\right)&\forall  k \in \cA\backslash \{d\}\\
1&\text{ if } k=d.
\end{cases}
\end{equation}
Then, by defining matrix $\bM$ of size  $(|\tA|\times|\tA|)$ with entries $M_{ka} = \mathbb{I}(a\in A(k))\exp\left(v(a|k)/\mu_k\right)$, $\forall a,k\in\tA$, $\bZ$ of size $|\tA|$ with entries
$Z_k = \exp(v(k)/\mu_k)$, the Bellman equation leads to the  system of \textit{nonlinear equations}
$Z_k = \sum_{a\in A(k)} M_{ka}Z_a^{\phi_{ka}} +b_k$, $\forall k\in \cA$. This system can be solved by a value iteration or least-squares approach \citep{MaiFrejinger22}. 

Recursive route choice models aforementioned  can be estimated via MLE, which requires to compute   log-likelihood values of some path observations. Given a complete path $\sigma = \{k_0,\ldots,k_{\cT} = d\}$, the probability of choosing this path can be computed as 
\[
P(\sigma|\btheta) = \prod_{t=0}^{\cT-1} P(k_{t+1}|k_t;\btheta),
\]
where $\btheta$ is the model parameters to be estimated and $P(k_{t+1}|k_t;\btheta)$ is the probability of moving from link $k_t$ to $k_{t+1}$. The link choice probabilities $P(a|k)$, $\forall a\in A(k)$ can be further computed using  the value function $V(k)$, $\forall k\in \cA$. For instance, for the RL model, we can compute  $P(a|k)$ as
\[
P(a|k) = \exp\left(\frac{1}{\mu}\left(v(a|k) + V(a) - V(k)\right)\right),
\]
and for the NRL, the link choice probabilities can be computed as
\[
P(a|k) = \exp\left(\frac{1}{\mu_k}\left(v(a|k) + V(a) - V(k)\right)\right). 
\]
In  the case that  path observation $\sigma$ is not complete, i.e., there is at least one pair of successive link observations $(k_t,k_{t+1})$ such that $k_{t+t} \notin A(k_t)$, the computation of $P(\sigma)$ is much more challenging. To elaborate this point, let consider a link choice probability $P(a|k)$ where link $a$ is not connected to $k$, i.e., $a\notin A(k)$. Basically,  $P(a|k)$ is referred to the probability of reaching link $a$ from $k$ conditional on the current parameters and the destination $d$. This probability can be written as
\begin{equation}\label{eq:missinglinks-directP}
    P(a|k) = \sum_{\sigma \in \Omega^{ka}} P(\sigma),
\end{equation}
where $\Omega^{ka}$ is the set of all possible paths going from $k$ to $a$ and $P(\sigma)$ is the probability of (complete) path $\sigma$, which can be computed through link choice probabilities of connected links. Enumerating all paths from $k$ to $a$ is typically not feasible, in particular with real networks. So, a direct approach to deal with missing data is practical.
In the next sections, we present our ways to address this issue, i.e., one is based on the  EM scheme \citep{dempster1977maximum}, and the other one is a direct approach to compute the log-likelihood function of incomplete observations.   


\section{The Expectation-Maximization Algorithm}\label{sec:EM}
We explore the EM method \citep{dempster1977maximum,arcidiacono2011conditional} to handle the data missing issue. Such an EM algorithm will alternate  an expectation (E) step, which creates  a function for the expectation of the log-likelihood (LL) evaluated using the current estimate for the parameters $\btheta$, and a maximization (M) step that maximizes the expectation function created by the E step. In the following, we show how this method can be applied to estimate  recursive route choice models with incomplete observations.

We first denote by $\Phi = \{\sigma_1,\ldots,\sigma_N\}$ the set of $N$ observed paths. This set may contain some incomplete paths. We further denote by $\bGm = \{\Gamma_1,\ldots,\Gamma_N\}$ the random variables representing the hidden parts of the observations. Each  $\Gamma_n$ corresponds to  $\sigma_n$ such that if $\gamma_n$ is a realization of $\Gamma_n$, then   $(\sigma_n,\gamma_n)$ forms a complete path. Given a realization $\bgm = \{\gamma_1,\ldots,\gamma_N\}$ of $\bGm$, we define $L(\btheta; \Phi,\bgm)$ be the LL function of the complete observations $\Phi \cup \bgm$. To perform the two EM steps, we  define $R(\btheta|\bar{\btheta})$ as the expected value of the LL function of  the observations w.r.t the current conditional distribution of $\bGm$ given the observed set $\Phi$ and the current estimate $\bar{\btheta}$ of the model parameters $\btheta$. The function $R(\btheta|\bar{\btheta})$ has the following form
\begin{equation}\label{eq:R}
\begin{aligned}
 R(\btheta|\bar{\btheta}) &=
 \bE_{\bGm|\Phi,\bar{\btheta}}\left[L(\btheta; \Phi,\bGm)\right] \\
 &= \sum_{n\in[N]} \bE_{\Gamma_n|\sigma_n,\bar{\btheta}}\left[\ln P( \sigma_n,\Gamma_n|\btheta)\right]. 
\end{aligned}
\end{equation}
Since $(\sigma_n,\gamma_n)$ forms a complete path for all $n\in [N]$, $P( \sigma_n,\gamma_n|\btheta)$ can be written as a product of choice probabilities of connected links. To compute the expectation in \eqref{eq:R}, we consider $\sigma_n$ as a set of pairs of links and 
we separate each $\sigma_n$ into two sets of pairs of links $\sigma^1_n$ and $\sigma^2_n$ such that $\sigma^1_n$ contains all the connected pairs of links, i.e., $\sigma^1_n = \{(g,h)\in\sigma_n|\ h\in A(g)\} $  and $\sigma^2_n =  \sigma_n\backslash \sigma^1_n$. This allows us to write \eqref{eq:R} as  
\begin{align}
 R(\btheta|\bar{\btheta})
 &= \sum_{n\in[N]} \bE_{\Gamma_n|\sigma_n,\bar{\btheta}}\left[\sum_{\varsigma \in \sigma^1_n}\ln P( \varsigma|\btheta) + \sum_{\varsigma \in \Gamma_n}\ln P( \varsigma|\btheta)\right]\nonumber\\
 & = \sum_{n\in[N]} \sum_{\varsigma \in \sigma^1_n}\ln P( \varsigma|\btheta) + \sum_{n \in [N]} \bE_{\Gamma_n|\sigma_n,\bar{\btheta}}\left[ \sum_{\varsigma \in \sigma_n^2}\ln P( \varsigma|\btheta)\right]\nonumber\\
 & = \sum_{n\in [N]} \sum_{\varsigma \in \sigma^1_n}\ln P( \varsigma|\btheta) + \sum_{n\in [N]} \sum_{\varsigma \in \sigma^2_n} \bE_{\Gamma^{n,\varsigma}|\;\bar{\btheta}}\big[\ln P( \Gamma^{n,\varsigma}|\btheta)\big],\label{eq:R-eq3}
\end{align}
where $\Gamma^{n,\varsigma}$ is a random variable representing a complete path connecting the pair of unconnected pair $\varsigma\in \sigma^2_n$,    $ \bE_{\Gamma^{n,\varsigma}|\varsigma,\bar{\btheta}}\big[\ln P( \Gamma^{n,\varsigma}|\btheta)\big]$ is an expectation over $\Gamma^{n,\varsigma}$, and $P(\Gamma^{n,\varsigma}|\btheta)$ is the probability of complete  path $\Gamma^{n,\varsigma}$. Since $\Gamma^{n,\varsigma}$ is complete,   $P(\Gamma^{n,\varsigma}|\btheta)$ can be written as the product of  choice probabilities of connected links.

An EM algorithm will perform alternatively E and M steps until converging to a fixed-point solution. At an E step and under a current estimate $\btheta^t$, we need to create function $R(\btheta|\btheta^t)$ based on $\Phi$ and the current parameters $\btheta^t$. Even though there may be infinite numbers of complete paths $\Gamma$ connecting a pair of links $\varsigma \in \sigma^2_n$, we can approximate the expectation by sampling as
\begin{equation}\label{eq:approx-R}
    \hat{R}(\btheta|\btheta^t) =
    \sum_{n\in [N]} \sum_{\varsigma \in \sigma^1_n}\ln P( \varsigma|\btheta) + \sum_{n\in [N]} \sum_{\varsigma \in \sigma^2_n} \sum_{s\in [S]} \ln P( \gamma^{n,\varsigma}_s|\btheta) P(\gamma^{n,\varsigma}_s|\varsigma,\btheta^t)/\cP^{n,\varsigma}(\btheta^t), 
\end{equation}
where $\gamma^{n,\varsigma}_1,\ldots,\gamma^{n,\varsigma}_S$ are $S$ realizations of $\Gamma^{n,\varsigma}$, $P(\gamma^{n,\varsigma}_s|\varsigma,\btheta^t)$ is the probability of the path  $\gamma^{n,\varsigma}_s$  given a pair of links $\varsigma$ and conditional on the current estimate $\btheta^t$, and $\cP^{n,\varsigma}(\btheta^t)=\sum_{s\in [S]} P(\gamma^{n,\varsigma}_s|\varsigma,\btheta^t)$.
Here, for ease of notation, we assume that the numbers of samples $S$ are the same over $\varsigma$. If it is not the case, we always can generate some dummy paths with zero probabilities.  
Typically, when the sample size $S$ is large enough, the left-hand side of \eqref{eq:approx-R} will converge to the true $R(\btheta|\btheta^t)$ with probability one due to the \textit{law of large number}.
 
To sample from the conditional distribution of $\Gamma^{n,\varsigma}$
given a pair $\varsigma = (k,a)$, $a\notin A(k)$, and parameters $\btheta^t$,
we can generate paths that connect these two links according to the link probabilities $P(h|g;\btheta^t)$, $\forall h\in A(g)$. The probability $P(\gamma^{n,\varsigma}|\;\btheta^t)$  can be computed by multiplying the link choice probabilities $P(h|g;\btheta^t)$ for every connected pair $(g,h)$ on $\gamma^{n,\varsigma}$.   In summary, we describe the EM algorithm in Algorithm \ref{alg:EM} below.

\begin{algorithm}[htb]
\DontPrintSemicolon
	\SetKwRepeat{Do}{do}{while}
\caption{\textit{EM algorithm}\label{alg:EM}}  Choose a converge threshold $\xi>0$, set $t=0$, choose an initial parameter vector $\btheta^0$.\; 
\Do{$||\btheta^{t}-\btheta^{t+1}|| \geq \xi$ \tcc{not converged to a stationary point}}{
\begin{itemize}
    \item \textbf{E step:} For each observation $\sigma_n$ and  for each unconnected pair $\varsigma \in \sigma^2_n$, generate paths $\gamma^{n,\varsigma}_1,\ldots,\gamma^{n,\varsigma}_S$ connecting that pair according to the link choice probabilities $P(h|g;\btheta^t)$, $\forall h\in A(g)$. Compute $P(\gamma^{n,\varsigma}_s|\btheta^t)$, $s \in [S]$. Create function $\hat{R}(\btheta|\btheta^t)$ using \eqref{eq:approx-R}.
    \item \textbf{M step:} Choose next $\btheta^{t+1}$ by maximizing $\widehat{R}(\btheta|\btheta^t)$
    \begin{equation}\label{eq:max-R}
    \btheta^{t+1} = \text{argmax}_{\btheta} \widehat{R}(\btheta|\btheta^t).    
    \end{equation}
\end{itemize}
$t\leftarrow t+1$
}
\end{algorithm}

The EM algorithm iterative performs the \textbf{E} and \textbf{M} steps and  stops when it converges to a fixed point solution, i.e., $\btheta^t$ is sufficiently closed to $\btheta^{t+1}$. It generates a sequence $\{\btheta^1,\btheta^2,\ldots\}$ such that the marginal likelihood of the data $L(\Phi|\btheta) = \int L(\btheta;\Phi,\Gamma)d\Gamma$ is non-decreasing. In general, the algorithm  does not guarantee that the sequence $\{\btheta^t\}$ converges to a maximum likelihood estimator, but typically, the derivative of the likelihood function will be arbitrarily close to zero at the fixed point \citep{dempster1977maximum,arcidiacono2011conditional}.

The \textbf{M} step of the EM algorithm involves solving a continuous maximization problem  $\max_{\btheta}\hat{R}(\btheta|\btheta^t)$. The computation of $\hat{R}(\btheta|\btheta^t)$ and its gradient w.r.t. $\btheta$ can be done by computing  the value function $V$ and its Jacobian. As mentioned earlier, these  value functions  can be obtained via  solving systems of linear equations for the RL model, and via value iteration for the NRL model. In the context of RL, it can be shown that $\widehat{R}(\btheta|\btheta^t)$ can be written as a sum of the logarithm of the choice probabilities of several complete paths. As a result, if the link utilities $v(a|k)$, $\forall a,k\in\tA$ are linear in $\btheta$, then  $\widehat{R}(\btheta|\btheta^t)$
is concave in $\btheta$, making the maximization problem in Algorithm \ref{alg:EM} tractable, noting that this linear-in-parameters setting is common in most of the existing route choice modeling studies \citep{Prato2012,FosgFrejKarl13,MaiFosFre15}.
We state this result in the following proposition.
\begin{proposition}
Under the RL model, if the deterministic link utilities $v(a|k)$, $k,a\in\tA$, $a\in A(k)$, are linear in parameters $\btheta$, then $\widehat{R}(\btheta|\btheta^t)$ is concave in $\btheta$. 
\end{proposition}
\proof{}
For notational simplicity,  let us define $p^{n,\varsigma}_s = P(\gamma^{n,\varsigma}_s|\varsigma,\btheta^t)/\cP^{n,\varsigma}(\btheta^t)$, for all $n \in [N]$, $s\in [S]$, and $\varsigma\in \sigma^2_n$. For any $n,\varsigma$ we have $\sum_{s\in [S]} p^{n,\varsigma}_s = 1$.
We now assume that each $\sigma^2_n$ contains $K_n$ disconnected pairs, denoted as $\varsigma^n_1,\ldots,\varsigma^n_{K_n}$ and let
$$\cS^n = \underbrace{[S]\times \ldots\times [S]}_{K_n\text{ times}},$$
recalling that $[S] = \{1,2,\ldots,S\}$.  
We now write $\widehat{R}(\btheta|\btheta^t)$ as
\begin{align}
    \widehat{R}(\btheta|\btheta^t) &=  \sum_{n\in [N]} \sum_{\varsigma \in \sigma^1_n}\ln P( \varsigma|\btheta) + \sum_{n\in [N]} \sum_{\varsigma \in \sigma^2_n} \sum_{s\in [S]} \ln P( \gamma^{n,\varsigma}_s|\btheta) p^{n,\varsigma}_s \label{eq:th1-eq1} 
\end{align}  
We now see that, for any $\varsigma^n_i\in \sigma^2_n$, $i\in [K_n]$,  and $s\in [S]$, the term $\ln P( \gamma^{n,\varsigma^n_i}_s|\btheta) p^{n,\varsigma^n_i}_s$ as
\begin{align}
    p^{n,\varsigma^n_i}_s \ln P( \gamma^{n,\varsigma^n_i}_s|\btheta)  &=  \ln P( \gamma^{n,\varsigma^n_i}_s|\btheta) \left(\prod_{\substack{j\in[K_n]\\j\neq i }} \left(\sum_{s'\in [S]} p_{s'}^{n,\varsigma^n_j}\right) \right) p^{n,\varsigma^n_i}_s\nonumber \\
    &=\ln P( \gamma^{n,\varsigma^n_i}_s|\btheta) \left(\sum_{\substack{(s_1,...,s_{K_n})\in \cS^n\\s_i = s}} \left( \prod_{j\in [K_n]} p^{n,\varsigma^n_j}_{s_j}\right)\right). \label{eq:th1-eq2}
\end{align}
Moreover, we have
\begin{equation}
\label{eq:th1-eq3}
\sum_{\varsigma \in \sigma^1_n}\ln P( \varsigma|\btheta) = \sum_{\varsigma \in \sigma^1_n}\ln P( \varsigma|\btheta) \left(\sum_{\substack{(s_1,...,s_{K_n})\in \cS^n}} \left( \prod_{j\in [K_n]} p^{n,\varsigma^n_j}_{s_j}\right)\right).
\end{equation}
We now can combine \eqref{eq:th1-eq1}-\eqref{eq:th1-eq3}
 to have
\begin{align}
  \widehat{R}(\btheta|\btheta^t)  &= \sum_{n\in [N]} \sum_{(s_1,...,s_{K_n})\in \cS^n}
    \left(
    \left(\prod_{i\in [K_n]} p^{n,\varsigma^n_i}_{s_i}\right)
    \left(
    \sum_{\varsigma \in \sigma^1_n}\ln P( \varsigma|\btheta) +   \sum_{i=1}^{K_n} \ln P( \gamma^{n,\varsigma^n_i}_{s_i}|\btheta)\right) 
    \right)\nonumber \\
&=    \sum_{n\in [N]} \sum_{(s_1,...,s_{K_n})\in \cS^n}
    \left(\prod_{i\in [K_n]} p^{n,\varsigma^n_i}_{s_i}\right)
    \left(\ln P\left((\sigma_n,\gamma^{n,\varsigma^n_1}_{s_1},...,\gamma^{n,\varsigma^n_{K_n}}_{s_{K_n}} )|\btheta\right)
    \right), \nonumber
\end{align}
where $P\left((\sigma_n,\gamma^{n,\varsigma^n_1}_{s_1},...,\gamma^{n,\varsigma^n_{K_n}}_{s_{K_n}} )|\btheta\right)$ is the probability of the (complete) path $(\sigma_n,\gamma^{n,\varsigma^n_1}_{s_1},...,\gamma^{n,\varsigma^n_{K_n}}_{s_{K_n}} )$.  \cite{FosgFrejKarl13} show that the choice probability of a complete path can be written as 
\[
P\left((\sigma_n,\gamma^{n,\varsigma^n_1}_{s_1},...,\gamma^{n,\varsigma^n_{K_n}}_{s_{K_n}} )|\btheta\right) = \frac{\exp(v(\sigma_n,\gamma^{n,\varsigma^n_1}_{s_1},...,\gamma^{n,\varsigma^n_{K_n}}_{s_{K_n}} |\btheta))}{\sum_{\tau \in \Omega^n} \exp(v(\tau|\btheta))},
\]
where $v(\tau|\btheta)$ is the deterministic utility of path $\tau$, i.e., $v(\tau) = \sum_{(k,a)\in\tau} v(a|k)$ and $\Omega^n$ is the set of all possible paths  that share the same destination with observation $n${-th}. Under the assumption that $v(a|k)$ is linear in $\btheta$, $v(\tau)$ is also linear in $\btheta$. As a result, we can write 
\begin{equation}
\label{eq:prop1-eq1}
\ln P\left((\sigma_n,\gamma^{n,\varsigma^n_1}_{s_1},...,\gamma^{n,\varsigma^n_{K_n}}_{s_{K_n}} )|\btheta\right) = v(\sigma_n,\gamma^{n,\varsigma^n_1}_{s_1},...,\gamma^{n,\varsigma^n_{K_n}}_{s_{K_n}} |\btheta) - \ln \left(\sum_{\tau \in \Omega^n} \exp v(\tau|\btheta)\right)
\end{equation}
The first term of \eqref{eq:prop1-eq1} is linear in $\btheta$ and the second term has a log-sum-exp convex form of a geometric program, thus is convex in $\btheta$. So, $\ln P\left((\sigma_n,\gamma^{n,\varsigma^n_1}_{s_1},...,\gamma^{n,\varsigma^n_{K_n}}_{s_{K_n}} )|\btheta\right)$ is concave, thus $\widehat{R}(\btheta|\btheta^t)$ is concave, as desired. 
\endproof

If the choice model is the NRL, the LL and $\widehat{R}(\btheta|\btheta^t)$ functions have highly nonlinear forms and are generally not concave. In this case, the M step of the EM algorithm may be difficult to handle globally. Moreover, it can be seen that the EM algorithm is generally expensive to perform, as it requires solving the maximization problem \eqref{eq:max-R} several times until getting a fixed point solution. Furthermore, each maximization problem \eqref{eq:max-R} requires sampling several paths between every unconnected pair of links  to approximate the expectation ${R}(\btheta|\btheta^t)$, which is  also  expensive.

\section{The Decomposition-Composition Algorithm }
\label{sec:DC}

In this section, we propose  a practical way to exactly compute  the likelihood of incomplete observations. We then present our \textit{decomposition-composition} method to speed up the computation of the LL function as well as its gradients. 

\subsection{Computing the Log-Likelihood of Incomplete Observations}

Given an unconnected pair $(k,a)$ such that the link observations between these two links are missing. 
To compute the LL of the path that contain pair $(k,a)$, we need to compute $P(a|k)$, i.e.,  the probability of reaching $a$ from $k$. 
As mentioned earlier, this can be done by enumerating all possible paths between these two links, which would be not possible with  a real network.
To do this in a tractable way, we define $\pi^a(s)$ as the probability of reaching link $a\in\tA$ from link $s\in\cA$. The values of $\pi^a(s)$, $\forall s,a\in \tA$, satisfies the following recursive equations 
\begin{equation}\label{eq:recursive-pi}
    \pi^a(s) = 
    \begin{cases}
    \sum_{s'\in A(s)}P(s'|s)\pi^a(s') & \forall s\neq a\\
    1 & s=a.
    \end{cases}
\end{equation}
So, if we define $\bh^a$ as a vector of size $|\tA|$ with all zero elements except $h^a_{a} = 1$ and  a matrix $\bQ^a$ of size $(|\tA|\times|\tA|)$ with elements 
\[
\bQ^{a}_{s,s'} =  
\begin{cases}
P(s'|s),& \forall s,s'\in \tA,\;  s\neq a, s'\in A(s) \\
0 &\text{ if } s = a,
\end{cases}
\]
then we can write \eqref{eq:recursive-pi} as 
\begin{equation}\label{eq:pi-linear-system}
\bpii^a = \bQ^{a} \bpii^a + \bh^a, \text{ or }\bpii^a = (\bI-\bQ^{a})^{-1}\bh^a,    
\end{equation}
where $\bpii^a$ is a vector of size $(|\tA|)$ with entries $\pi^a_s = \pi^a(s)$, $\forall s\in\tA$, and $\bI$ is the identity matrix of size $(|\tA|\times|\tA|)$. Proposition \ref{theor:invertible-IQ} below shows that $\bI-\bQ^{a}$ is invertible, which guarantees the existence and uniqueness of the solutions to \eqref{eq:pi-linear-system}.

\begin{proposition}\label{theor:invertible-IQ}
For any $a\in\tA$, the matrix $\bI-\bQ^{a}$ is invertible. 
\end{proposition}
\begin{proof}
For any $a\in \tA$, we have
\[
\sum_{s'\in \tA} \bQ^a_{s,s'} =  
\begin{cases}
 \sum_{s'\in \tA} P(s'|s) = 1 & \text{ if } s\neq a \\
 0 & \text{ if }s= a.
\end{cases} 
\]
So, $\bQ^a$ is a \textit{sub-stochastic} matrix (contains non-negative entries and every row adds up to at most 1). Moreover, $\bQ^a$ contains no  recurrent class. So, $(\bI-\bQ^a)$ is invertible, as desired.  
\end{proof}

So after computing $\bpii^a$,  we can obtain $P(a|k)$ as 
$P(a|k) = \bpii^a_{k}$. We also need the gradients of $P(a|k)$ for the maximum likelihood estimation.  Such gradient values can be obtained through taking the Jacobian of $\bpii^a$ w.r.t. parameter $\theta_j$ as
$$
\frac{\partial \bpii^a}{\partial \theta_j} = -(\bI-\bQ^{a})^{-1}\frac{\partial \bQ^a}{\partial \theta_j}(\bI-\bQ^{a})^{-1}\bh^a = -(\bI-\bQ^{a})^{-1}\frac{\partial \bQ^a}{\partial \theta_j} \bpii^a,
$$
where $\partial \bQ^a/\partial \theta_j$ is a matrix of size $|\tA|\times |\tA|$ with entries $\partial Q^a_{s,s'}/\partial \theta_j$. This matrix can be obtained by taking gradients of $P(k'|k)$ for all $k',k\in \tA,k'\in A(k)$. These further can be computed by taking the gradients of the value functions, and the gradients of the value function $V(k)$ can be computed by solving some system of linear equations \citep{FosgFrejKarl13,MaiFosFre15}. So, the Jacobian of $\bpii^a$ can be obtained by solving a series of  \textit{systems of linear equations}.  


\subsection{Decomposition-Composition Algorithm}
In the above section, we show that  the LL of incomplete observations can be computed directly by solving one system of linear equations for each unconnected pair $(k_i,k_{i+1})$. This would be time-consuming if a path observation contains a large number of such unconnected pairs. We show in the following  that it is possible to 
obtain the probabilities of all the unconnected pairs on paths that  share the same destination by solving only one system of linear equations, instead of solving one linear system per each incomplete pair. The idea is to \textit{decompose} each individual system of linear equations into two parts among which one part is the same over unconnected pairs. This allows us to \textit{compose} all the individual linear systems into only one linear system.

For a path observation, assume that we observe $J$ incomplete pairs $\{(u_1,v_1)$,..., $(u_J,v_J)\}$, $(u_j,v_j) \in \tA\times \tA$, $\forall j=1,\ldots,J$. We define a matrix $\bQ^0$ of size $((|\tA|+1)\times (|\tA|+1)$ with entries
\[
\begin{aligned}
\bQ^0_{s,s'} &=  P(s'|s),\ \forall s,s'\in\tA, s'\in A(s)\\
\bQ^0_{s,|\tA|+1} &= \bQ^0_{|\tA|+1,s} = 0,\  \forall s\in\A.
\end{aligned}
\]
That is, the last row and last column of $\bQ^0$ are all-zero vectors. {We also define  
 a matrix $\bH$ of size $(|\tA|+1)\times K$ with entries}
\[
H_{s,k} = \begin{cases}
1  &\text{ if } s=|\tA|+1 \text{ or } s = u_k \\
0 &\text{ otherwise. }
\end{cases}
\]
The following theorem show that one can  obtain all the probabilities $P(v_j|u_j)$, $j=1,\ldots,J$ by solving only \textit{one system of linear equations}. Let we first denote by  $\widetilde{\bI}$ as an identity matrix of size $(|\tA|+1)\times (|\tA|+1)$. 
\begin{theorem}
\label{theo:th2} 
The matrix $\widetilde{\bI} - \bQ^0$  is always  invertible. Moreover,  if $\pmb{\Pi}$ is the unique solution to the system of linear equations $\pmb{\Pi} = (\widetilde{\bI}-\bQ^0)^{-1}\bH$, then 
$
P(v_j|u_j)= {\Pi}_{v_j, j},\; \forall j \in [J].
$
\end{theorem}
\proof{}
To prove the invertibility, Let $\widehat{\bQ} = ({\bQ}^0)^\transpose$, we can see that $\widehat{Q}_{s,s'}\geq 0 $ for all $s,s'\in \tA$ and
$
\sum_{s'\in A(s)} \widehat{Q}_{s,s'} = 1, \;\forall s\in \tA
$ and $\sum_{s'\in A(s)} \widehat{Q}_{s,s'}=0 $ if $s = |\tA|+1$. Thus $\widehat{\bQ}$ is a \textit{sub-stochastic matrix}. Moreover, $\widehat{\bQ}$ contains no recurrent class, thus $(\widetilde{\bI} - \widehat{\bQ})$ is invertible. Moreover $(\widetilde{\bI} - \widehat{\bQ})^\transpose = (\widetilde{\bI} - {\bQ}^0)$, thus $\widetilde{\bI} - \bQ^0$ is invertible, as desired.  

For the second claim of the theorem,  let us consider a pair of missing data $(u_j,v_j)$. We create an artificial link $r$ such that it is impossible to reach $r$ from any other link in $\tA$ and  from $r$ we can only reach $u_k$ with probability 1. Let $\pi^j(s)$, for all $s\in \A$ be the probability of reaching $s$ from $u_j$, i.e., $P(s|u_j)$. We see that $\pi^j(s)$, $s\in \A$, satisfy the following recursive equations.
\begin{equation}\label{eq:recursive-pi'}
    \pi^j(s) = 
    \begin{cases}
    \sum_{s'\in \cA}P(s|s')\pi^j(s') & \forall s\in \A\\
    1 & s=r.
    \end{cases}
\end{equation}
Now, we define a matrix $\bQ^j$ of size $(|\tA|+1) \times (|\tA|+1)$ with elements
$\bQ^j_{s,s'} =  P(s|s')$, $\forall s,s'\in \tA\cup\{r\}, s\in A(s')$,
and a vector $\bq$ of size $|\tA|+1$ with all zero elements except $q_r = 1$. We also number  $r$ as $|\tA|+1$, so the last column and row of $\bQ^j$ and the last element of $\bq$ correspond to link $r$.  \eqref{eq:recursive-pi'} can be written as 
\begin{equation}\label{eq:linear-system-pik}
\bpii^j = \bQ^j \bpii^j+\bq.    
\end{equation}
Note that the last row of $\bQ^j$ is an all-zero vector as $r$ is not reachable by any other links, and the last column of $\bQ^j$ has all zero entries except the element that corresponds to link $u_j$. 
We  decompose \eqref{eq:linear-system-pik} as
\[\begin{aligned}
\bQ^j\bpii^j+\bq  &=
\begin{pmatrix}
\bQ^j_{1,1} & \bQ^j_{1,2} & \cdots & \bQ^j_{1,|\tA|+1} \\
\bQ^j_{2,1} & \bQ^j_{2,2} & \cdots & \bQ^j_{2,|\tA|+1} \\
\vdots  & \vdots  & \ddots & \vdots  \\
0 & 0 & \cdots & 0
\end{pmatrix}
\begin{pmatrix}
\pi^j_1\\
\pi^j_2\\
\vdots\\
\pi^j_{|\tA|+1}
\end{pmatrix} +\bq \\
&\stackrel{(a)}{=} \begin{pmatrix}
0 & 0 & \cdots & \bQ^j_{1,|\tA|} &0 \\
0 & 0 & \cdots & \bQ^j_{2,|\tA|} & 0 \\
\vdots  & \vdots  & \ddots & \vdots &\vdots   \\
0 & 0 & \cdots & \bQ^j_{2,|\tA|} & 0
\end{pmatrix}\bpii^j + \begin{pmatrix}
0 & 0 & \cdots & \bQ^j_{1,|\tA|+1} \\
0 & 0 & \cdots & \bQ^j_{2,|\tA|+1} \\
\vdots  & \vdots  & \ddots & \vdots  \\
0 & 0 & \cdots & 0
\end{pmatrix}
\begin{pmatrix}
\pi^j_1\\
\pi^j_2\\
\vdots\\
\pi^j_{|\tA|+1}
\end{pmatrix} +\bq \\
&=\bQ^0 \bpii^j+\begin{pmatrix}
0 & 0 & \cdots & \bQ^j_{1,|\tA|+1} \\
0 & 0 & \cdots & \bQ^j_{2,|\tA|+1} \\
\vdots  & \vdots  & \ddots & \vdots  \\
0 & 0 & \cdots & 0
\end{pmatrix}
\begin{pmatrix}
\pi^j_1\\
\pi^j_2\\
\vdots\\
\pi^j_{|\tA|+1}
\end{pmatrix} +\bq \\
&\stackrel{}{=} \bQ^0 \bpii^j+ \begin{pmatrix}
\bQ^j_{1,|\tA|+1}\\
\bQ^j_{2,|\tA|+1}\\
\vdots\\
0
\end{pmatrix} \pi^j_{|\tA|+1} + \bq 
\end{aligned}
\] 
where for \textit{(a)} we separate $\bQ^j$ into two parts; the first part contains all the columns of $\bQ^j$ except the last column and the second part only contains the last column of $\bQ^j$. We now note that  the last column of $\bQ^j$ have zero element except the one corresponding to $u_j$. Moreover $\pi^j_{|\tA|+1} = 1$, thus $(Q^j_{1,|\tA|+1},\ldots, Q^j_{|\tA|+1,|\tA|+1})^\transpose \pi^j_{|\tA|+1}$ is a vector of zero elements except the one at $u_k$, which is equal to 1. Thus, 
\[
\begin{pmatrix}
\bQ^j_{1,|\tA|+1}\\
\bQ^j_{2,|\tA|+1}\\
\vdots\\
0
\end{pmatrix} \pi^j_{|\tA|+1} + \bq
\]
is also the $j$-th column of matrix $\bH$ defined above. So, we have
\[
\bQ^j\bpii^j + \bq = \bQ^0\bpii^j + \bH_{:,j},
\]
where $\bH_{:,k}$ is the $j$-th column of $\bH$. So, $\bpii^j$ is a solution to the following system
\[
(\widetilde{\bI}-\bQ^0)\bpii^j = \bD_{:,j},
\]
where $\widetilde{\bI}$ is an identity matrix of size $(|\tA|+1)\times (|\tA|+1)$. This also means that if matrix $\pmb{\Pi}$ is a solution to the system of linear equations $(\widetilde{\bI}-\bQ^0)\pmb{\Pi} = \bD$ or  $\bpi = (\widetilde{\bI}-\bQ^0)^{-1}\bH$, then
\[
\bpii^j = \pmb{\Pi}_{:,j}\text{ and } {\Pi}_{v_j,j} = \bpii^j(v_j) = P(v_j|u_j),
\]
which is the desired result.
\endproof

For the MLE,  the gradients of $P(v_j|u_j)$ are required. The gradient of $P(v_j|u_j)$ w.r.t.  a parameter $\theta_t$ can also be obtained by computing the Jacobian of $\bpi$, which can be obtained by solving  the following linear system
\begin{equation}\label{eq:gradient-missing-prob}
    \frac{\partial \pmb{\Pi}}{\partial \btheta_t} = (\widetilde{\bI}-\bQ^0)^{-1}\frac{\partial \bQ^0}{\partial \btheta_t}(\widetilde{\bI}-\bQ^0)^{-1} \bH = (\widetilde{\bI}-\bQ^0)^{-1}\frac{\partial \bQ^0}{\partial \btheta_t} \bpi.
\end{equation}
Since $\widetilde{\bI}-\bQ^0$ is invertible, the linear system in \eqref{eq:gradient-missing-prob}  always has a unique solution.

Algorithm \ref{algo:LL-missing} below  describes the basic steps of our Decomposition-Composition (DC) algorithm (the term ``Decomposition-Composition'' refers to the fact that  we decompose matrix  $\bQ^j$ into two parts where one of them is the same over unconnected pairs $(u_j,v_j)$, and then compose all the individual  linear systems  into 
only one system of linear equations). 
\begin{algorithm}[htb]
	\tcc{Run a nonlinear optimization algorithm (gradient  accent or trust region or linear search \citep{NoceWrig06}) to estimate the model parameters  $\btheta$.}
	\While{\textit{not converged}}{
	\begin{itemize}
	    \item[(i)] Compute the value function $\bZ$ as well as its gradients by value iteration and/or solving systems of linear equations, as shown in \cite{FosgFrejKarl13,MaiFosFre15}
	    \item[(ii)] Compute link choice probabilities $P(a|k)$, for all $k,a\in\tA, a\in A(k)$ and their gradients
	    \item[(iii)] Solve the linear system $\bpi =  (\widetilde{\bI}-\bQ^0)^{-1}\bH$ and \eqref{eq:gradient-missing-prob} to get the probabilities of unconnected pairs $(u_j,v_j)$, $j = 1,\ldots,J$, as well as their gradients.
	    \item[(iv)] Compute the log-likelihood of the (incomplete) observations and its gradient w.r.t. $\btheta$
	    \item[(v)] Update $\btheta$
	\end{itemize} 
	}
	\caption{(\textbf{\textit{DC algorithm}})}\label{algo:LL-missing}
\end{algorithm}

The total number of linear systems to be solved in Algorithm \ref{algo:LL-missing} to compute the probabilities of unconnected  links is $(T+1)N^{\textsc{dest}}$, where $N^{\textsc{dest}}$  is the number of destinations  and $T$ is the size of $\btheta$. Here we note that if the model specification involves an origin-destination specific attribute, e.g., the link size attribute \citep{FosgFrejKarl13}, then  the total number of linear systems to be solved depends on the number of origin-destination pairs $N^{\textsc{od}}$ in the observations, i.e.,  $(T+1)N^{\textsc{od}}$. 
 Clearly, the  number of linear systems to be solved does not depend on the number of incomplete pairs of links in the observations.
Moreover, Algorithm \ref{algo:LL-missing} can be implemented in a parallel manner, which would help to speed up the computation.  
If the data is complete (no missing segment), then we just need to remove Step (iii) from the algorithm. 

It should be noted that Step (iii) of Algorithm \ref{algo:LL-missing} requires to solve some systems of linear equations that all involve the matrix $\bI-\bQ^0$. The $\bL\bU$ factorization is also a convenient approach to achieve good performance. Technically speaking, one can firstly decompose $\bI-\bQ^0$ into a lower triangular matrix $\bL$ and an upper  triangular matrix $\bU$ and use these two matrices to solve the corresponding linear system.

To perform Step (iii) of Algorithm \ref{algo:LL-missing},  we need to compute $\partial P(a|k)/\partial \theta_i$ for any parameter $\theta_i$, for any $a\in A(k)$.
 $P(a|k)$ can be computed using the utility $v(a|k)$ and value function $V(a)$ and $V(k)$. For instance, for the RL model, $P_{ka} = \exp\left(\frac{1}{\mu}(v(a|k)+V(a)-V(k))\right)$. By taking the logarithm of $P_{ka}$ and taking the derivatives on the both sides, the derivative of $P(a|k)$ w.r.t  a parameter $\theta_i$ can be easily obtained, for the RL model \citep{FosgFrejKarl13}, as 
\[
\frac{\partial P_{ka}}{\partial \theta_i} = \frac{P_{ka}}{\mu}\left( \frac{\partial v(a|k)}{\partial \theta_i}+\frac{\partial V(a)}{\partial \theta_i} - \frac{\partial V(k)}{\partial \theta_i} \right),
\]
For the NRL model \citep{MaiFosFre15}, the derivatives have more complicated formulas,  as we assume that the scales $\mu_k$ are also functions of the model parameters $\btheta$.
\[
\frac{\partial P_{ka}}{\partial \theta_i} = P_{ka}\left(\frac{1}{\mu_k} \left(\frac{\partial v(a|k)}{\partial \theta_i}+\frac{\partial V(a)}{\partial \theta_i} - \frac{\partial V(k)}{\partial \theta_i} \right) - \frac{1}{\mu_k^2}\frac{\partial \mu_k}{\partial \theta_i} \left(v(a|k)+V(a)-V(k)\right)\right).
\]
Moreover, as shown in previous work \citep{MaiFosFre15,FosgFrejKarl13}, the Jacobian of $V$ can be computed conveniently by solving a system of linear equations.

As discussed above, we have argued that the EM algorithm would be expensive as it requires several MLE problems until getting a fixed point solution. On the other hand, the DC algorithm requires to solve only one MLE problem, but at each iteration, we need to solve several systems of linear equations (or dynamic programs) to obtain the probabilities of unconnected pairs of links. So, it is to be expected that each iteration of the DC algorithm would be less expensive than an iteration of the EM, but more expensive than an iteration of the classical nested fixed-point algorithm used in previous work for the case of complete data \citep{MaiFrejinger22}.   

\section{Numerical Experiments}
\label{sec:experiments}

In this section, we present our numerical results with the RL \citep{FosgFrejKarl13} and the NRL models \citep{MaiFosFre15}. We use a dataset collected in Borl\"ange, Sweden, which has been used in several prior route choice studies \citep{FosgFrejKarl13,MaiBasFre15_DeC,MaiFosFre15}. The trip observations in the data are complete (there is no unconnected link observation) and to evaluate how the proposed algorithms handle the issue of data incompleteness, we  will randomly remove links from the  full observations.   

\subsection{Data and Experimental Settings}
We describe the dataset and our experimental settings in the following. The dataset is  based on GPS observations of car trajectories collected in Borl\"ange, Sweden (a network composed of 3,077 nodes and 7,459 links). The sample consists of 1,832 trips corresponding to simple paths with a minimum of five links. Moreover, there are 466 destinations, 1,420 different origin-destination (OD) pairs and more than 37,000 link choices in this sample.
We use the RL model specification of \cite{FosgFrejKarl13} and the  NRL model specification  of \cite{MaiFosFre15}. The same four attributes as in \citep{FosgFrejKarl13,MaiFosFre15} are used for the instantaneous utilities, i.e., link travel time $TT(a)$ of link $a$, left turn dummy $LT(a|k)$ that equals one if the turn angle from $k$ to $a$ is larger than 40 and less than 177,  U-turn dummy $UT(a|k)$ that equals one if the turn angle is larger than 177, and  link constant $LC(a)$ that is equal to 1 for every link $a\in \A$. For the scales $\mu_k$  of the NRL model, we also use three attributes  as in \cite{MaiFosFre15} for the scale specification $\mu_k$, which are travel time (TT), Link Size (LS) and the number of outgoing links (OL). 


We  generate datasets of missing information by removing  links from the complete trip observations. To evaluate the performance of our proposed algorithms (i.e. the EM and DC algorithms), we will run them with the generated datasets and use the parameter estimates obtained to compute   LL values based on the set of complete  observations. We will compare the EM and DC algorithms with two baselines  based on the  standard estimation algorithm used in previous work, i.e.  the nested-fixed point (NFXP) algorithm \citep{RUST1987}. The first baseline is based on the complete dataset and the second one is only based on the connected links in the generated datasets. More precisely, for the second baseline, given an incomplete  trip observation $\sigma  = \{k_1,\ldots,k_L\}$, we compute the  LL function based on $P(k_{i+1}|k_i)$ if  $k_i$  and $k_{i+1}$ are connected, i.e., $k_{i+1} \in A(k_i)$, for $i\in [L-1]$. The difference between the two proposed algorithms (EM and DC) and the second baseline is that the EM and DC algorithms take into consideration the  probabilities of unconnected links, whereas the baselines ignore them. As a result, it is  expected that the two baselines would run faster than the EM and DC algorithms. We denote the first baseline by NFXP-C (standing for the NFXP algorithm with \textit{complete} data) and the second baseline by NFXP-I (standing for the NFXP algorithm with \textit{incomplete} data). 

We  generate datasets of missing information by going through all the trip observations in the dataset and  randomly removing links (except the origins and destinations) with a given probability $p$. That is, given a removing probability $p \in [0,1]$, we remove each link observation  with a probability $p$. We vary $p$ from 0.1 to 0.9, where small $p$ values imply that the missing information is low, and large $p$ values mean that the missing information is high. For each removing probability $p$, 10 seeds are created. Thus, there are 10 incomplete datasets  generated for each removing probability $p$. In total there are 90 datasets  generated.

 We  estimate the RL and NRL models using the  EM, DC, NFXP-I algorithms with each generated incomplete dataset, and the NFXP-C with the complete dataset (the one before link removals). We then use the parameter estimates to compute the LL values of the complete dataset. Since the objective is to maximize the LL of the complete observations, the higher the LL values, the better the algorithm is. This will allow us to evaluate how the proposed algorithms (EM and DC) and the baseline NFXP-I recover the missing information, as compared to the NFXP-C which works with the complete dataset and always returns the highest LL values.  
 
 For the EM algorithm, we select the convergence threshold as $\xi = 10^{-4}$ and the number of path samples for each unconnected pair of links as $S=5$. Here we note that a higher number of samples can be chosen but since the EM  algorithm is expensive to execute, especially when the number of path samples is large (as we will show later), we restrict ourselves to $S=5$.   
 
The algorithms are coded in Python using Pytorch library \citep{paszke2019pytorch}. The experiments are conducted  on a PC with processor  Intel(R) Xeon(R) Silver 4116 CPU @ 2.10GHz with 64Gb RAM, equipped with a GPU NVIDIA Quadro P1000.
We use the limited memory BFGS (L-BFGS) method \citep{Nocedal2006} for the MLE problems (i.e., searching over the parameter space of $\btheta$ for the DC algorithm and performing the  \textit{M step} of the EM algorithm). This is a Quasi-Newton method that computes the Jacobian matrix updates with gradient evaluations, but only saves some last  updates  to save memory.

\subsection{Numerical Comparison}
In this section, we present numerical comparisons of the proposed algorithms (EM and DC) and the two baselines (NFXP-I and NFXP-C), based on the RL and NRL models. We first compare the LL values given by the four approaches, with a note that NFXP-C is based on the complete data and is considered as ground truth, thus always gives the highest LL values. The higher LL values imply that the corresponding algorithm performs better in recovering missing information from incomplete datasets.

\begin{figure}[htb]
    \centering
    \includegraphics[width=0.9\linewidth]{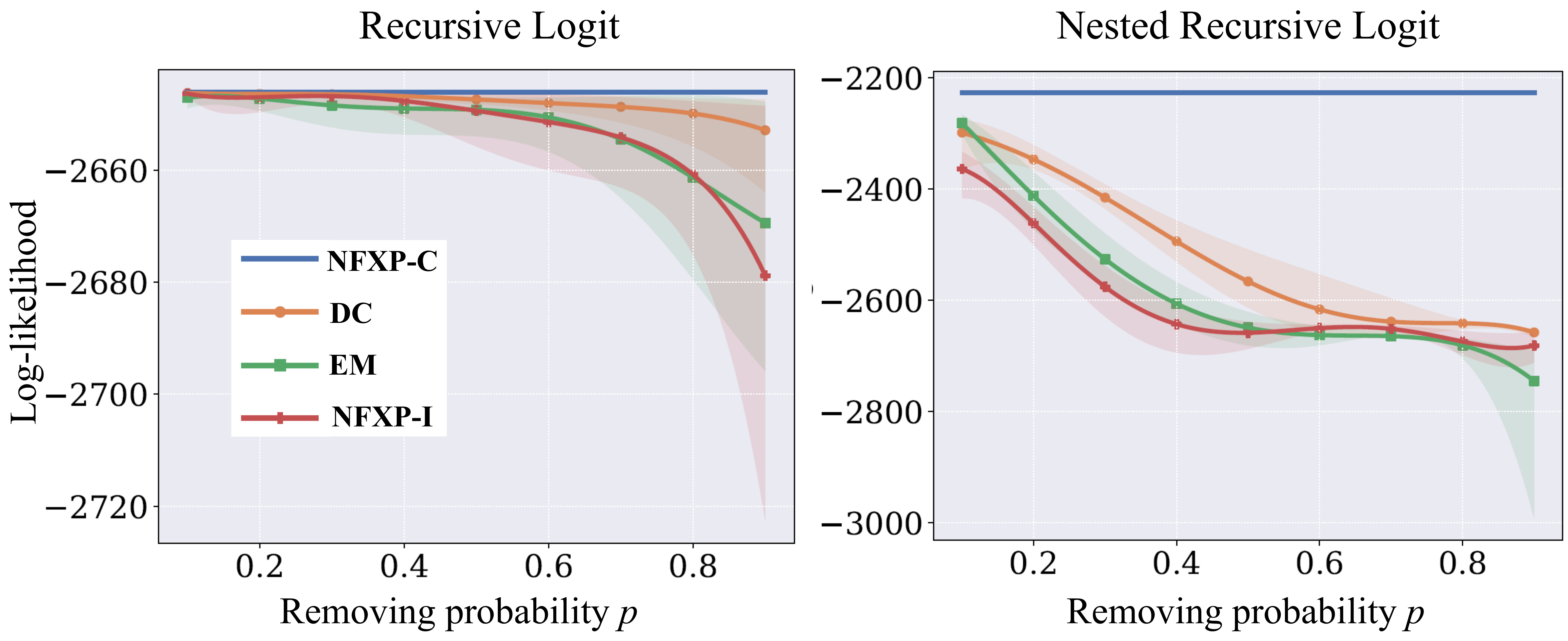}
       \caption{Log-likelihood comparison for the RL and NRL models.}
    \label{fig:LL-comparison}
\end{figure}    

Figure \ref{fig:LL-comparison} reports the means and standard errors of the  LL values given by the four approaches, where the solid curves represent the mean values,  and the  shaded areas represent the standard errors. We can see that, for both RL and NRL models, the DC always gives higher LL values, as compared to those given by the EM and NFXP-I. The EM seems to be better than NFXP-I in some cases,  and worse in some other cases. For the RL model, all the approaches (DC, EM and NFXP-I)  seem to perform well  as  $p$ increases. The gaps  between the DC, EM and NFXP-I are small for $p\leq 0.6$ and become more significant  when $p\geq 0.7$. This demonstrates the efficiency of the DC  algorithm in handling highly incomplete data.
For the NRL model,  we however see that the DC, EM and NFXP-I seem to perform worse, as compared to the case of the RL model; the LL values drop more  quickly as $p$ increases. This is to be expected as we know that the NRL model has more complicated forms and has more information to infer (i.e., the parameters of the link utilities and  the scales $\mu_k$, $k\in \tA$). We however still see that the DC significantly outperforms the EM and NFXP-I approaches for $p\leq 0.7$  and is slightly better than EM and NFXP-I for $p\geq 0.8$. Moreover, EM is slightly better than  NFXP-I for  $p\leq 0.5$ and slightly worse than NFXP-I for $p=0.9$.

\begin{table}[htb]
\centering
\begin{tabular}{l|c|c|c|c}
Model     & \begin{tabular}[c]{@{}c@{}}Removing\\ Probability\end{tabular} & \textbf{EM}         & \textbf{DC}        & \textbf{NFXP-I}     \\ 
\hline
\multirow{10}{*}{\textbf{RL}}  & 0.0& \multicolumn{3}{c}{-2646.08}\\ 
\cline{3-5}
& 0.1& -2646.98$\pm$0.73   & -2646.25$\pm$0.18  & -2646.43$\pm$0.28   \\
& 0.2& -2647.71$\pm$1.51   & -2646.40$\pm$0.21  & -2646.82$\pm$0.84   \\
& 0.3& -2647.77$\pm$1.25   & -2646.53$\pm$0.37  & -2647.00$\pm$0.72   \\
& 0.4& -2649.23$\pm$2.34   & -2646.84$\pm$0.36  & -2647.77$\pm$1.30   \\
& 0.5& -2649.57$\pm$2.74   & -2647.29$\pm$0.60  & -2649.31$\pm$2.81   \\
& 0.6& -2650.58$\pm$1.92   & -2648.23$\pm$0.76  & -2651.07$\pm$2.88   \\
& 0.7& -2654.02$\pm$6.61   & -2648.61$\pm$1.45  & -2654.83$\pm$7.00   \\
& 0.8& -2661.72$\pm$11.03  & -2649.97$\pm$2.42  & -2660.45$\pm$8.87   \\
& 0.9& -2669.37$\pm$17.59  & -2652.90$\pm$5.45  & -2678.95$\pm$24.05  \\ 
\hline
\multirow{10}{*}{\textbf{NRL}} & 0.0& \multicolumn{3}{c}{-2162.32}\\ 
\cline{3-5}
& 0.1& -2292.91$\pm$46.48  & -2258.39$\pm$62.34 & -2361.82$\pm$29.23  \\
& 0.2& -2452.83$\pm$42.77  & -2379.86$\pm$61.90 & -2493.52$\pm$23.33  \\
& 0.3& -2545.09$\pm$29.06  & -2419.17$\pm$28.45 & -2571.09$\pm$20.18  \\
& 0.4& -2615.28$\pm$20.70  & -2474.93$\pm$24.33 & -2625.09$\pm$22.23  \\
& 0.5& -2668.62$\pm$28.33  & -2528.01$\pm$19.49 & -2683.15$\pm$26.50  \\
& 0.6& -2728.13$\pm$22.93  & -2585.90$\pm$11.48 & -2732.94$\pm$27.95  \\
& 0.7& -2669.20$\pm$6.99   & -2631.07$\pm$17.03 & -2656.76$\pm$8.80   \\
& 0.8& -2676.56$\pm$8.86   & -2652.93$\pm$11.63 & -2668.96$\pm$10.79  \\
& 0.9& -2827.12$\pm$347.75 & -2659.35$\pm$5.79  & -2684.21$\pm$18.56 
\end{tabular}
\caption{Means and standard errors of the  LL values (over 10 runs) for the RL and NRL models.}
\label{tab:details-LLvalues}
\end{table}

We also report the details of the means and standard errors of the LL values in Table \ref{tab:details-LLvalues}, where the values before ``$\pm$'' are the means and the values after ``$\pm$'' are the standard errors. The LL values reported for $p=0.0$ are also the LL values given by the NFXP-C (i.e., given by the NFXP with the complete data). It can be seen that the LL values given by the NRL model are significantly higher than those given by the RL model. This is to be expected as previous studies already show that the NRL model always performs better than the RL in terms of in- and out-of-sample fits. Moreover, we see that the standards errors given by the RL are  much smaller than those given by the NRL model, for all the three algorithms, which can  be explained by the fact that the NRL model yields a more complicated choice probability formulation and is more difficult  to estimate. Furthermore, looking at the LL values across the three approaches EM, DC, and NFXP-I, we also see that DC does not only return larger mean LL values, but also yields small LL standard errors. This would be due to the fact that the DC approach computes the probabilities of unconnected segments exactly while the EM algorithm approximates them by sampling, and the NFXP-I simply ignores the missing segments. 
This observation demonstrates the efficiency and stability of the DC approach, as compared to the other ones. 

In summary,  our experiments show that DC  denominates the other approaches (i.e., EM and NFXP-I) in recovering  missing information, in the sense that the LL values given by the DC algorithm are always higher (thus are closer to the ground-truth values), as compared to those given by EM and NFXP-I. Moreover, our approaches seem to provide better results for the RL model than for the NRL model.

\begin{figure}[htb]
    \centering
    \includegraphics[width=0.9\linewidth]{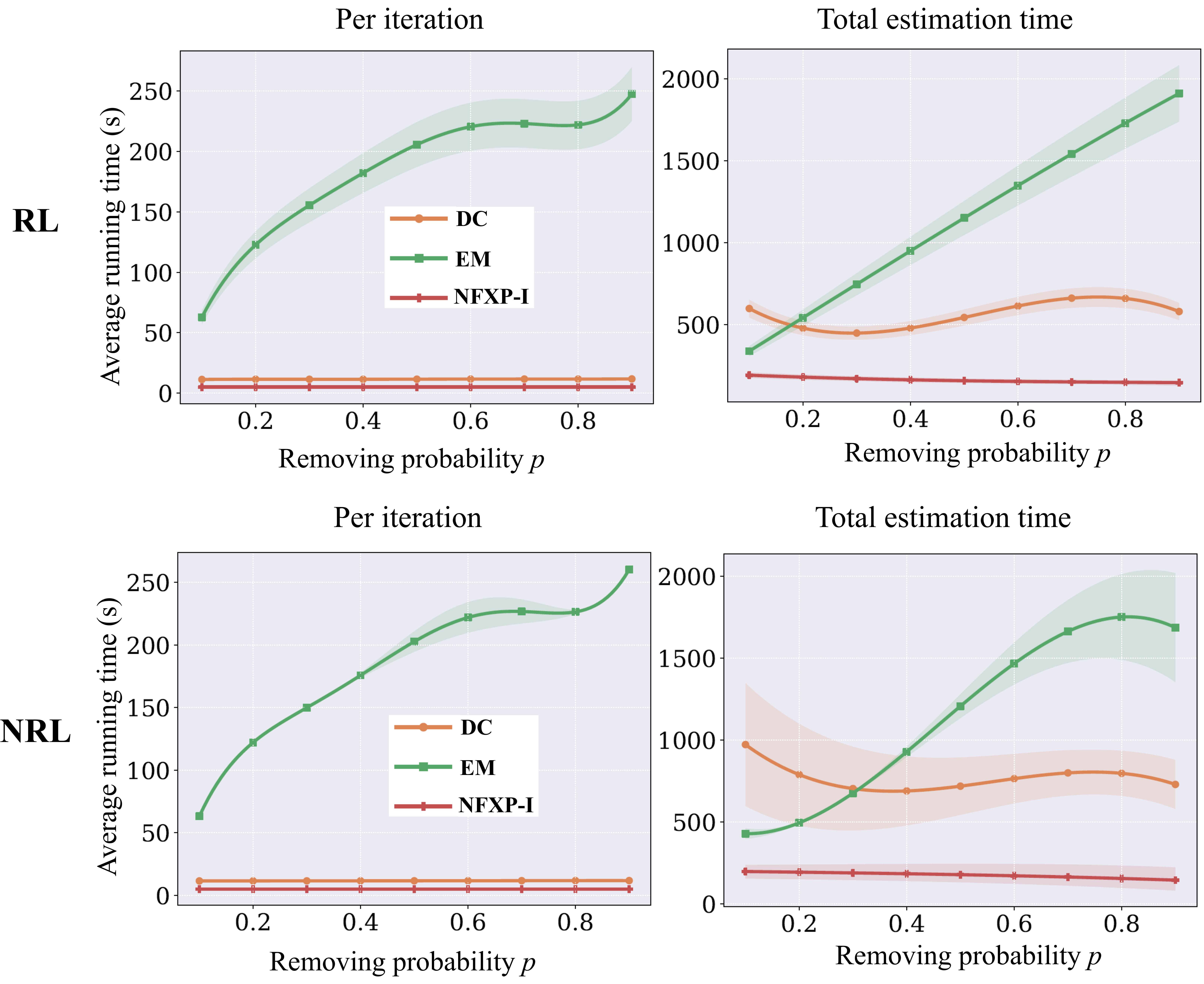}
       \caption{Average running time comparison for  the RL and NRL models.}
    \label{fig:CPUtime}
\end{figure}

We now turn our attention to the computational costs required to run our algorithms, which is also a focus of this work. In Figure \ref{fig:CPUtime} we plot the means (solid curves) and standard errors (shaded areas) of the running times of the three algorithm DC, EM, and NFXP-I, for  the RL and NRL models, noting that the running times of NFXP-I and NFXP-C are generally  similar.  In each figure, the left-hand side reports the average running times per iteration  and the right-hand  side reports the average total running times.  It can be seen that, per iteration, DC and EM are much faster than EM, and as expected, NFXP-I is faster than DC. Importantly, the running  times (per iteration) of EM grow  quickly as $p$ increases while  the running times (per iteration) of DC and NFXP-I are stable. This is because when the removing probability $p$ increases, the dataset would contain more unconnected segments and the EM algorithm would need more samples to approximate the expected LL function, thus becoming more expensive. In contrast, the DC algorithm is developed  in such a way that the number of systems of equations to be solved \textit{does not depend} on the number of unconnected segments in the path observations. The number of linear systems to be solved also does not change as $p$ increases for the NFXP-I algorithm. Thus, an increase  in $p$ would not make a significant impact on the running times of DC and NFXP-I.       

In terms of total estimation time, it is interesting to see that the EM algorithm is faster than  the DC for some small $p$ values (e.g., $p\leq 0.2$ for the RL and $p\leq 0.3$ for the NRL models), even though the running times per iteration of the EM are higher. This would be explained as follows. When $p$ is small, the number of unconnected segments in the observations is small, thus the EM algorithm only requires a small number of path samples and $\widehat{R}(\btheta|\btheta^t)$ would give a good approximation to the expected LL function. Moreover, the form of the approximate expected LL function in \eqref{eq:approx-R}, even though looks complex, it shares a similar structure with the standard LL functions studied in previous works \citep{FosgFrejKarl13,MaiFosFre15}. On the other hand, the LL function in the DC approach involves some exact probabilities of unconnected segments, thus would be more complicated. All these would be the reason for the fact that when $p$ is small, EM requires fewer iterations to converge and is faster, as compared to the DC algorithm. 

When $p\geq 0.2$ for the RL model and $p\geq 0.3$ for the NRL model, the running times of EM start becoming higher than those of DC and increase fast. This is to be expected as the running times per iteration of EM also increases fast as $p$ increase. The total estimation times of the DC algorithm, as expected, are stable and seem not to be affected much by an increase in $p$. Moreover, looking  across the two route choice models, we see that the total running times of DC are higher for the NRL model, as compared to those given by the RL model. This is to be expected as the NRL model is known to be more  expensive to estimate than the RL model. However, it seems not to be the case for the EM approach. The standard errors of the running times seem to be higher for the NRL model, which is consistent with our above observations for the LL comparisons. We refer the reader to the appendix for the details of the running times. 

In summary, in terms of running time,  the DC and EM  algorithms  are  more expensive to perform than the naive approach NFXP-I. The DC algorithm is generally much faster than  the EM, especially when $p$ is large. More importantly, the running times of EM grow fast as the data-missing level $p$ increases, while an increase in $p$ does not make a significant impact on the running times of DC. All these indicate the practical tractability of the DC approach for handling highly missing data.

\section{Conclusion}
\label{sec:concl}

In this paper, we have studied the issue of incomplete observations in route choice modeling. we have developed two solution approaches, namely, the EM and DC algorithms, aiming to efficiently recover missing information from data. While the EM algorithm requires to solve several MLE problems and would be expensive to perform, the DC algorithm is based on the idea that we can exactly compute the probabilities of unconnected segments by solving systems of linear equations. The main advantage of our DC algorithm is that the number of system of linear equations to be solved to compute the probabilities of unconnected segments does not depend on the number of unconnected segments, thus this approach scales  well when the number of unconnected segments increases. We provide numerical experiments based on a dataset collected from a real network, which showed that the DC algorithm outperforms the EM and other baselines in recovering  missing information. Our numerical results also showed that the DC algorithm is generally faster than the EM approach, and the running times of the DC seem not to be affected by the amount of missing information in the data observations.   



Our methods can be used  with other recursive models developed in the literature, e.g., the dynamic RL \citep{de2020RL_dynamic} or stochastic time-dependent RL \citep{mai2021RL_STD} models,  and may have applications beyond route choice modeling, for example, the estimation of  activity-based models or  general structural dynamic discrete choice models with missing data. These would shape some interesting directions for future work.    


\section*{Acknowledgements}

{This research is supported by  Singapore Ministry of Education (MOE) Academic Research Fund (AcRF) Tier 1 grant (Grant No: 20-C220-SMU-010) to the first and second authors}


\bibliographystyle{plainnat_custom}
\bibliography{refs}

\clearpage

\section*{Appendix}

Tables \ref{tab:CPU-time-RL} and \ref{tab:CPU-time-NRL}
below report the details of the running times (per iteration and overall) of the three algorithms (EM, DC and NFXP-I). The means are reported before $``\pm''$ and the standard errors are reported after $``\pm''$.
The running times of the NFXP-I are in general small and stable as $p$ increases. It is important to note that, in terms of  running time per iteration, the running time standard errors of the DC and NFXP-I are considerably smaller than those of the EM algorithm, but the overall running time standard errors of the DC and NFXP-I are comparable with those of the EM, indicating that the numbers of iterations required by DC and NFXP-I vary significantly across the 10 independent runs.

\begin{table}[htb]
\centering
\begin{tabular}{c|c|c|c|c}
\multicolumn{1}{l|}{}&\begin{tabular}[c]{@{}c@{}}Removing \\probability\end{tabular}  & EM                 & DC               & NFXP-I            \\ 
\hline
\multirow{9}{*}{Per iteration}         & 0.1  & 62.85$\pm$5.66     & 11.34$\pm$1.02   & 5.40$\pm$0.49     \\
& 0.2  & 102.16$\pm$9.19    & 11.49$\pm$1.03   & 5.61$\pm$0.51     \\
& 0.3  & 141.71$\pm$12.75   & 11.50$\pm$1.04   & 5.56$\pm$0.50     \\
& 0.4  & 169.21$\pm$15.23   & 11.48$\pm$1.03   & 5.46$\pm$0.49     \\
& 0.5  & 195.89$\pm$17.63   & 11.52$\pm$1.04   & 5.50$\pm$0.49     \\
& 0.6  & 214.90$\pm$19.34   & 11.61$\pm$1.04   & 5.60$\pm$0.50     \\
& 0.7  & 223.32$\pm$20.10   & 11.67$\pm$1.05   & 5.58$\pm$0.50     \\
& 0.8  & 221.64$\pm$19.95   & 11.67$\pm$1.05   & 5.42$\pm$0.49     \\
& 0.9  & 228.72$\pm$20.59   & 11.68$\pm$1.05   & 5.37$\pm$0.48     \\ 
\hline
\multirow{9}{*}{\begin{tabular}[c]{@{}c@{}}Total \\estimation time\end{tabular}} & 0.1  & 336.88$\pm$30.32   & 597.28$\pm$53.76 & 190.23$\pm$17.12  \\
& 0.2  & 452.45$\pm$40.72   & 517.06$\pm$46.54 & 183.30$\pm$16.50  \\
& 0.3  & 650.87$\pm$58.58   & 452.12$\pm$40.69 & 173.28$\pm$15.60  \\
& 0.4  & 848.78$\pm$76.39   & 457.00$\pm$41.13 & 165.35$\pm$14.88  \\
& 0.5  & 1061.45$\pm$95.53  & 511.65$\pm$46.05 & 158.76$\pm$14.29  \\
& 0.6  & 1255.07$\pm$112.96 & 580.52$\pm$52.25 & 154.18$\pm$13.88  \\
& 0.7  & 1460.70$\pm$131.46 & 644.85$\pm$58.04 & 150.49$\pm$13.54  \\
& 0.8  & 1645.63$\pm$148.11 & 667.59$\pm$60.08 & 147.91$\pm$13.31  \\
& 0.9  & 1824.82$\pm$164.23 & 629.08$\pm$56.62 & 145.82$\pm$13.12 
\end{tabular}
\caption{Means and standard errors of the running times (in seconds - over 10 runs), overall estimation time  and per iteration, for the RL model.}
\label{tab:CPU-time-RL}
\end{table}

\begin{table}[htb]
\centering
\begin{tabular}{c|c|c|c|c}
& \begin{tabular}[c]{@{}c@{}}Removing \\probability\end{tabular} & EM                 & DC                & NFXP-I            \\ 
\hline
\multirow{9}{*}{Per iteration}         & 0.1  & 63.17$\pm$1.03     & 11.53$\pm$0.00    & 5.63$\pm$0.36     \\
& 0.2  & 103.30$\pm$2.18    & 11.52$\pm$0.04    & 5.54$\pm$0.42     \\
& 0.3  & 138.25$\pm$0.43   & 11.54$\pm$0.06    & 5.54$\pm$0.43     \\
& 0.4  & 162.58$\pm$0.13   & 11.58$\pm$0.06    & 5.58$\pm$0.40     \\
& 0.5  & 191.02$\pm$5.32    & 11.61$\pm$0.06    & 5.59$\pm$0.38     \\
& 0.6  & 214.39$\pm$10.92   & 11.64$\pm$0.06    & 5.54$\pm$0.41     \\
& 0.7  & 226.40$\pm$11.60   & 11.69$\pm$0.06    & 5.48$\pm$0.48     \\
& 0.8  & 225.55$\pm$5.46    & 11.75$\pm$0.07    & 5.44$\pm$0.54     \\
& 0.9  & 235.50$\pm$-1.01   & 11.79$\pm$0.10    & 5.44$\pm$0.54     \\ 
\hline
\multirow{9}{*}{\begin{tabular}[c]{@{}c@{}}Total \\estimation time\end{tabular}} & 0.1  & 428.10$\pm$32.51   & 972.34$\pm$375.06 & 197.09$\pm$40.28  \\
& 0.2  & 448.51$\pm$14.83   & 854.86$\pm$337.80 & 194.84$\pm$43.20  \\
& 0.3  & 579.70$\pm$4.97    & 732.59$\pm$280.47 & 190.68$\pm$49.28  \\
& 0.4  & 797.22$\pm$17.67   & 688.93$\pm$231.93 & 186.04$\pm$56.14  \\
& 0.5  & 1082.86$\pm$52.50  & 701.55$\pm$189.90 & 180.32$\pm$63.69  \\
& 0.6  & 1348.92$\pm$100.02 & 741.97$\pm$161.46 & 174.27$\pm$70.04  \\
& 0.7  & 1591.04$\pm$163.63 & 787.37$\pm$142.52 & 166.72$\pm$75.46  \\
& 0.8  & 1728.20$\pm$229.48 & 804.20$\pm$136.61 & 158.69$\pm$78.41  \\
& 0.9  & 1740.45$\pm$298.27 & 771.47$\pm$142.39 & 149.52$\pm$78.70 
\end{tabular}
\caption{Means and standard errors of the running times (in seconds - over 10 runs), overall  and per iteration, for the NRL model.}
\label{tab:CPU-time-NRL}
\end{table}

\end{document}